\documentclass[runningheads,envcountsame]{llncs}
\usepackage{graphicx,xcolor}
\usepackage{amsmath,amssymb,latexsym,mathrsfs,latexsym}
\usepackage{enumitem}
\usepackage{xspace}
\usepackage{multicol}

\renewcommand{\phi}{\varphi}

\newcommand{\ab}[1]{\langle#1\rangle}
\renewcommand{\iff}{\text{iff}}
\newcommand{\Lra}{\Leftrightarrow}
\newcommand{\lra}{\leftrightarrow}

\newcommand{\mcp}{\mathcal{P}}

\newcommand{\msc}{\mathscr{C}}

\newcommand{\pr}{\textsc{prop}\xspace}
\newcommand{\qiff}{\quad\iff\quad}
\newcommand{\Ra}{\Rightarrow}
\newcommand{\ra}{\rightarrow}

\newcommand{\MCS}{\text{MCS}}
\newcommand{\modelEx}{\ensuremath{(S,\rel,V)}}

\newcommand{\sig}{\ensuremath{(\logic,\ax,\alpha,\iota)}\xspace}
\newcommand{\tail}{\mathsf{tail}}
\newcommand{\pointedmodel}{\ensuremath{(M,s)}\xspace}
\newcommand{\rel}{\ensuremath{R}}

\newcommand{\cdomain}{\ensuremath{\mathsf{S}}\xspace}
\newcommand{\cmodel}{\ensuremath{\mathsf{M}}\xspace}
\newcommand{\crelation}{\ensuremath{\mathrel{\mathsf{R}}}\xspace}
\newcommand{\cvaluation}{\ensuremath{\mathsf{V}}\xspace}

\newcommand{\lang}{\ensuremath{\mathcal{L}}\xspace}
\newcommand{\langcd}{\ensuremath{\mathcal{L}^{\cap\Union}}\xspace}
\newcommand{\langd}{\ensuremath{\mathcal{L}^\cap}\xspace}

\newcommand{\Int}{\ensuremath{\cap}\xspace}
\newcommand{\Union}{\ensuremath{\uplus}\xspace}
\newcommand{\Unionplus}{\text{Union$^+$}\xspace}
\newcommand{\unionplus}{\text{union$^+$}\xspace}

\newcommand{\ax}{\textbf{L}\xspace}
\newcommand{\axk}{\textbf{K}\xspace}
\newcommand{\axd}{\textbf{D}\xspace}
\newcommand{\axt}{\textbf{T}\xspace}
\newcommand{\axsfour}{\textbf{S4}\xspace}
\newcommand{\axb}{\textbf{B}\xspace}
\newcommand{\axsfive}{\textbf{S5}\xspace}

\newcommand{\axcdk}{\ensuremath{{\textbf{K}^{\cap\Union}}}\xspace}
\newcommand{\axcdd}{\ensuremath{{\textbf{D}^{\cap\Union}}}\xspace}
\newcommand{\axcdt}{\ensuremath{{\textbf{T}^{\cap\Union}}}\xspace}
\newcommand{\axcdsfour}{\ensuremath{{\textbf{S4}^{\cap\Union}}}\xspace}
\newcommand{\axcdb}{\ensuremath{{\textbf{B}^{\cap\Union}}}\xspace}
\newcommand{\axcdsfive}{\ensuremath{{\textbf{S5}^{\cap\Union}}}\xspace}

\newcommand{\axdk}{\ensuremath{{\textbf{K}^\cap}}\xspace}
\newcommand{\axdd}{\ensuremath{{\textbf{D}^\cap}}\xspace}
\newcommand{\axdt}{\ensuremath{{\textbf{T}^\cap}}\xspace}
\newcommand{\axdsfour}{\ensuremath{{\textbf{S4}^\cap}}\xspace}
\newcommand{\axdb}{\ensuremath{{\textbf{B}^\cap}}\xspace}
\newcommand{\axdsfive}{\ensuremath{{\textbf{S5}^\cap}}\xspace}

\newcommand{\ck}{\textbf{Un}\xspace}
\newcommand{\ckk}{\textbf{Un(\axk)}\xspace}
\newcommand{\ckd}{\textbf{Un(\axd)}\xspace}
\newcommand{\ckt}{\textbf{Un(\axt)}\xspace}
\newcommand{\cksfour}{\textbf{Un(\axsfour)}\xspace}
\newcommand{\ckb}{\textbf{Un(\axb)}\xspace}
\newcommand{\cksfive}{\textbf{Un(\axsfive)}\xspace}

\newcommand{\dkk}{\textbf{Int(\axk)}\xspace}
\newcommand{\dkd}{\textbf{Int(\axd)}\xspace}
\newcommand{\dkt}{\textbf{Int(\axt)}\xspace}
\newcommand{\dksfour}{\textbf{Int(\axsfour)}\xspace}
\newcommand{\dkb}{\textbf{Int(\axb)}\xspace}
\newcommand{\dksfive}{\textbf{Int(\axsfive)}\xspace}

\newcommand{\lk}{\text{K}\xspace}
\newcommand{\ld}{\text{D}\xspace}
\newcommand{\lt}{\text{T}\xspace}
\newcommand{\lsfour}{\text{S4}\xspace}
\newcommand{\lb}{\text{B}\xspace}
\newcommand{\lsfive}{\text{S5}\xspace}

\newcommand{\logic}{\text{L}\xspace}

\newcommand{\lcdk}{\ensuremath{\text{K}^{\cap\Union}}\xspace}
\newcommand{\lcdd}{\ensuremath{\text{D}^{\cap\Union}}\xspace}
\newcommand{\lcdt}{\ensuremath{\text{T}^{\cap\Union}}\xspace}
\newcommand{\lcdsfour}{\ensuremath{\text{S4}^{\cap\Union}}\xspace}
\newcommand{\lcdb}{\ensuremath{\text{B}^{\cap\Union}}\xspace}
\newcommand{\lcdsfive}{\ensuremath{\text{S5}^{\cap\Union}}\xspace}

\newcommand{\ldk}{{\ensuremath{\text{K}^{\cap}}}\xspace}
\newcommand{\ldd}{{\ensuremath{\text{D}^{\cap}}}\xspace}
\newcommand{\ldt}{{\ensuremath{\text{T}^{\cap}}}\xspace}
\newcommand{\ldsfour}{{\ensuremath{\text{S4}^{\cap}}}\xspace}
\newcommand{\ldb}{{\ensuremath{\text{B}^{\cap}}}\xspace}
\newcommand{\ldsfive}{{\ensuremath{\text{S5}^{\cap}}}\xspace}

\newcommand{\indexes}{\mathcal{I}}

\begin{document}
\title{Simpler completeness proofs for modal logics with intersection}
%
%
\author{Y\`{i} N. W\'{a}ng\inst{1}
\and
Thomas {\AA}gotnes\inst{2,3}
}
%
%
\institute{Zhejiang University, Hangzhou, China\\\email{ynw@xixilogic.org}
\and University of Bergen, Bergen, Norway
\and Southwest University, Chongqing, China\\
\email{Thomas.Agotnes@infomedia.uib.no}}
\maketitle              
\begin{abstract}
  There has been a significant interest in extending various modal
  logics with intersection, the most prominent examples being
  epistemic and doxastic logics with distributed knowledge.
  Completeness proofs for such logics tend to be complicated, in
  particular on model classes such as S5 like in standard epistemic
  logic, mainly due to the undefinability of intersection of
  modalities in standard modal logics. A standard proof method for the
  S5 case was outlined in \cite{HM1992} and later explicated in more
  detail in \cite{WA2013pacd}, using an ``unraveling-folding method''
  case to achieve a treelike model to deal with the problem of
  undefinability. This method, however, is not easily adapted to other
  logics, due to the level of detail and reliance on S5. In this paper we
  propose a simpler proof technique by building a treelike canonical
  model directly, which avoids the complications in the processes of
  unraveling and folding. We demonstrate the technique by showing
  completeness of the normal modal logics K, D, T, B, S4 and S5
  extended with intersection modalities. Furthermore, these treelike
  canonical models are compatible with Fischer-Ladner-style closures, and
  we combine the methods to show the completeness of the mentioned
  logics further extended with transitive closure of union
  modalities known from PDL or epistemic logic. Some of these
  completeness results are new.

\end{abstract}
\section{Introduction}

Intersection plays a role in several areas of modal logic, including
epistemic logics with distributed knowledge \cite{MvdH1995,FHMV1995},
propositional dynamic logic with intersection of
programs \cite{HKT2000}, description logics with concept intersection
\cite{BCMNP-S2017,BHLS2017}, and coalition logic \cite{asia2015}.  It
is well-known that intersection is not modally definable and that
standard logics with intersection are not canonical (cf., e.g.,
\cite{hoek92making}).

A method for proving completeness was introduced by
\cite{HM1992,hoek92making,FHMV1995,MvdH1995} for various (static)
epistemic logics with distributed knowledge, and later explicated and
extended in \cite{WA2013pacd,WA2015rck-del,Wang2013thesis} as the
\emph{unraveling-folding method} which is applicable to various static
or dynamic epistemic S5 logics with distributed knowledge with or
without common knowledge.

Let us take a closer look at this technique for epistemic logic with
distributed knowledge (S5D). It is known that the canonical S5 model
built in the standard way is not a model for the classical
axiomatization for this logic. This is because the accessibility
relation $R_G$ (where $G$ is a set) that is (implicitly) used to
interpret the interesection (distributed knowledge) modality is not
necessarily the intersection of individual accessibility relations
$R_a$ ($a \in G$).  In the canonical S5 model we can ensure that
${R_G}\subseteq \bigcap_{a\in G} {R_a}$, but not that
${R_G} \supseteq \bigcap_{a\in G} {R_a}$.

The unraveling-folding method is carried out in the following way. A 
\emph{pre-model} is a standard S5 model where $R_G$ is treated as a
primitive relation for each group $G$. A \emph{pseudo model}
is a pre-model satisfying the following two constraints
\begin{enumerate}
\item ${R_{\{a\}}} = {R_a}$ for every agent $a$, and 
\item ${R_G} \subseteq {\bigcap_{a\in G} {R_a}}$ for every agent $a$ and group $G$
\end{enumerate}
An S5D model is then a pseudo model that satisfies also a third constraint:
\begin{enumerate}[resume]
\item ${R_G} \supseteq {\bigcap_{a\in G} {R_a}}$ for every agent $a$ and group $G$
\end{enumerate}
A canonical pseudo model can be truth-preservingly translated to a
\emph{treelike} pre-model using an \emph{unraveling} technique, and
then \emph{folded} to an S5D model while also preserving the truth of
all formulas (for details of the two processes see
\cite{WA2013pacd}). Completeness is achieved by first building a
canonical pseudo model for a given consistent set $\Phi$ of formulas,
and then having it translated to an S5D model for $\Phi$ using the
unraveling-folding method.

There are many subtleties not mentioned in this simplified overview,
which in particular makes the method cumbersome to adapt to extensions
of basic epistemic logic or to non-S5 based logics. 

In this paper we demonstrate a simpler way to prove completeness for
modal logics with intersection. Since we know that a treelike model typically
works for such logics, the idea is to build a treelike model
\emph{directly} for a given consistent set of formulas. We call
such a model a \emph{standard model}. This eliminates having to deal
with the details of the unraveling and folding processes, and
dramatically simplifies the proofs.

We illustrate the technique by building the standard model for each of
the modal logics, K, D, T, B, S4 and S5, extended with
intersection. We furthermore demonstrate that the method is useful by
showing that it is compatible with finitary methods based on
Fischer-Ladner-style closures, and introduce finitary standard models
for the mentioned logics further extended with the transitive closure
of the union, using in, e.g., PDL and epistemic logic, as well. Some
of these completeness results have been stated in the literature
before, often without proof, some of them not.  For example, to the
best of our knowledge, no completeness results have been reported for
D or B extended with intersection, and even less can be found for
logics with both intersection and the transitive closure of union.

The rest of the paper is structured as follows. In the next section we
introduce basic definitions and conventions. We study some modal logics
with (only) intersection in Section \ref{sec:logicd}, introduce an
axiomatization for each of them and show its completeness, and then
study the logics extended further with transitive closure of union in Section
\ref{sec:logiccd}. We conclude in Section \ref{sec:con}.

\section{Preliminaries}
\label{sec:pre}

In this paper we study modal logics over multi-modal languages with
countably many standard unary modal operators: $\Box_0$, $\Box_1$,
$\Box_2$, etc. On top of these we focus on two types of modal
operators, each \emph{indexed} by a finite nonempty set $I$ of natural numbers:
\begin{itemize}
\item \emph{Intersection modalities}, denoted $\Int_I$;
\item \emph{\Unionplus modalities}, denoted $\Union_I$.
\end{itemize}
We mention some applications of these modalities below.

The languages are parameterized by a countably infinite set $\pr$ of
propositional variables, and an at-most countable set $\indexes$ of primitive types. A
finite non-empty subset $I \subseteq \indexes$ is called an \emph{index}.  We are interested in the following languages.

\begin{definition}[languages]
$$\begin{array}{ll}
(\lang)& \phi ::=  p \mid \neg \phi \mid (\phi\ra\phi) \mid \Box_i\phi\\
(\langd)& \phi ::= p \mid \neg \phi \mid (\phi\ra\phi) \mid \Box_i\phi \mid \Int_I\phi \\
(\langcd)&\phi ::= p \mid \neg \phi \mid (\phi\ra\phi) \mid \Box_i\phi \mid \Int_I\phi \mid \Union_I\phi\\
\end{array}$$
where $p\in\pr$, $i\in\indexes$ and $I$ an index. Other Boolean connectives are defined as usual.
\end{definition}

A Kripke model $M$ (over $\pr$ and $\indexes$) is a triple \modelEx, where $S$ is a nonempty set of states, $\rel:\indexes\to\wp(S\times S)$ assigns to every modality $\Box_i$ a binary
relation $\rel_i$ on $S$, and $V:\pr\to S$ is a valuation which associates with every propositional variable a set of states where it is true.

\begin{definition}[satisfaction]
For a given formula $\alpha$, the \emph{truth} of it in, or its \emph{satisfaction} by, a model $M=\modelEx$ with a designated state $s$, denoted $M,s\models\alpha$, is defined inductively as follows.
$$\begin{array}{l@{\quad}l@{\quad}l}
M,s\models p&\iff&s\in V(p)\\
M,s\models \neg\phi&\iff&\text{not}\ \pointedmodel\models\phi\\
M,s\models(\phi\ra\psi)&\iff&M,s\models\phi\ \text{implies}\ M,s\models\psi\\
M,s\models \Box_i\phi&\iff&\text{for all $t\in S$, if $(s,t) \in \rel_i$ then $M,t\models\phi$}\\[.5ex]
M,s\models \Int_I\phi&\iff&\text{for all $t\in S$, if $(s,t)\in \bigcap_{i\in I}\rel_i$ then $M,t\models\phi$}\\[.5ex]
M,s\models \Union_I\phi&\iff&\text{for all $t\in S$, if $(s,t) \in \biguplus_{i\in I} \rel_i$ then $M,t\models\phi$}
\end{array}$$
where\footnote{Although the symbol $\biguplus$ is sometimes used for
  disjoint union, we repurpose it here for transitive closure of
  union.} $\biguplus_{i\in I} R_i$ is the transitive closure of $\bigcup_{i\in I} R_i$.
\end{definition}

Thus, the \emph{intersection} modalities are interpreted by taking the
intersection, and the \emph{\unionplus} modalities by taking the
transitive closure of the union. We use ``\unionplus
modalities'' as a short name to avoid the more awkward ``transitive closure of
union modalities''.

Given a formula $\phi$ and a class $\msc$ of models, we say $\phi$ is
\emph{valid} in $\msc$ iff $\phi$ is valid in all models of $\msc$. We
usually do not choose a class of models arbitrarily, but are rather
interested in those based on a certain set of conditions over the
binary relations in a model. Such conditions are often called
\emph{frame conditions}. In this paper we are going to focus on some
of the most well known frame conditions (see, e.g., \cite{Chellas1980}). These conditions are \emph{seriality}, \emph{reflexivity},  \emph{symmetry}, \emph{transitivity} and \emph{Euclidicity}. It is well known that these frame conditions are characterized by the formulas D ($\Box_i\phi\ra\neg \Box_i\neg\phi$), T ($\Box_i\phi\ra\phi$), B ($\neg\phi\ra \Box_i\neg \Box_i\phi$), 4 ($\Box_i\phi\ra \Box_i\Box_i\phi$) and 5 ($\neg \Box_i\phi\ra \Box_i\neg \Box_i\phi$), respectively. With respect to different combinations of these frame conditions, normal modal logics \lk, \ld (a.k.a. KD), \lt (a.k.a. KT), \lb (a.k.a. KTB), \lsfour (a.k.a., KT4) and \lsfive (a.k.a. KT5) based on the language \lang are well studied in the literature. We shall refer an ``S5 model'' to a Kripke model in which the binary relation is an equivalence relation, and likewise for a D, T, B or S4 model.

In this paper we will focus on the counterpart logics over the languages \langd and \langcd, and they will be named in a comprehensive way as follows:
\begin{center}
\begin{tabular}{llllll}
\ldk,& \ldd,& \ldt, & \ldb,& \ldsfour,  &\ldsfive,\\
\lcdk,& \lcdd,& \lcdt, & \lcdb,& \lcdsfour,  &\lcdsfive.\\
\end{tabular}
\end{center}

There are well known applications of these logics, for example are \ldsfive and
\lcdsfive (under the restriction that $\indexes$ is finite) well known
as S5D and S5CD respectively in the area of epistemic logic. The
logics \ldk and \ldsfour are known as $\mathcal{ALC}(\Int)$ (i.e.,
$\mathcal{ALC}$ with role intersection) and $\mathcal{S}(\Int)$ (where
$\mathcal{S}$ is $\mathcal{ALC}$ with role transitivity) respectively
in the area of description logic \cite{BCMNP-S2017,BHLS2017}.%
\footnote{The subscript $i$ of a unary modal operator $\Box_i$
  typically stands for an agent in epistemic logic or a role in
  description logic. In epistemic logic, a finite number of
  agents is assumed, and the intersection modality (i.e., a
  distributed knowledge operator) is an arbitrary intersection over a
  finite domain. In
  description logic, the number of roles are typically unbounded,
  but the intersection is binary, which is in effect equivalent to finite
  intersection. }  The logic \lcdk is
close to propositional dynamic logic with intersection (IPDL)
\cite{HKT2000} or the description logic $\mathcal{ALC(\cap,\cup,*)}$,
and similarly, \lcdsfour close to $\mathcal{S(\cap,\cup,*)}$.%
\footnote{There are two major differences however. First, the Kleene
  star in both logics are the reflexive-transitive closure, and we
  consider the transitive closure which is denoted by a ``$+$'' in the
  symbol $\Union$. Second, $\Union_{I}$ is a compound modality (union
  and then take the transitive closure), while in those logics the
  Kleene star is separated from the union, and as a result, the Kleene
  star applies to the intersection as well, which we do not consider
  here.}

The minimal logic K can be axiomatized by the system \axk composed of the following axiom (schemes) and rules (where $\phi,\psi\in\lang$ and $i\in\indexes$):
\begin{itemize}[leftmargin=3em]
\item[(PC)] all instances of all propositional tautologies
\item[(MP)] from $(\phi\ra\psi)$ and $\phi$ infer $\psi$
\item[(K)] $\Box_i(\phi\ra\psi) \ra (\Box_i\phi \ra \Box_i\psi)$
\item[(N)] from $\phi$ infer $\Box_i\phi$
\end{itemize}
Axiomatizations for D, T, B , S4 and S5, which are named \axd, \axt,
\axb, \axsfour and \axsfive respectively, can be obtained by adding
characterization axioms to \axk. In more detail, $\axd = \axk \oplus
\text{D}$, $\axt = \axk \oplus \text{T}$, $\axb = \axt \oplus
\text{B}$, $\axsfour = \axt \oplus \text{4}$ and $\axsfive = \axt
\oplus \text{5}$, where the symbol $\oplus$ means combining the axioms and rules of the two parts. Details can be found in standard modal logic textbooks (see, e.g., \cite{Chellas1980,BdRV2001}).

A logic extended with the intersection modality is typically
axiomatized by adding axioms and rules to the corresponding logic
without intersection. The axioms and rules to be added are in total
called the \emph{characterization of intersection}, and depends on
which logic we are dealing with. Similarly we can define the \emph{characterization of the transitive closure of union}, which can be made independent to the concrete logic (will be made clear in Section \ref{sec:logiccd}).

Characterizations of intersection and transitive closure of union can be found in the literature for some of the logics, including \ldk, \ldt, \ldsfour, \ldsfive and \lcdsfive in epistemic logic (see \cite{FHMV1995,MvdH1995,Wang2013thesis}). In particular, for the base logic S5, the characterizations are \dksfive and \cksfive, respectively:

\paragraph{\dksfive} characterization of intersection in \axdsfive and
\axcdsfive (D\Int, 4\Int, B\Int and N\Int are not necessary in the
sense that they are derivable):
\begin{itemize}
\item (K\Int) $\Int_I(\phi\ra\psi)\ra (\Int_I\phi \ra \Int_I\psi)$
\item (D\Int) $\Int_I \phi \ra \neg \Int_I \neg \phi$
\item (T\Int) $\Int_I \phi \ra \phi$
\item (4\Int) $\Int_I \phi \ra {\Int_I} {\Int_I} \phi$
\item (B\Int) $\neg\phi \ra {\Int_I} \neg {\Int_I} \phi$
\item (5\Int) $\neg{\Int_I}\phi \ra {\Int_I} \neg {\Int_I} \phi$
\item (N\Int) from $\phi$ infer $\Int_I\phi$
\item ($\Int$1) $\Box_i\phi \lra \Int_{\{i\}}\phi$
\item ($\Int$2) $\Int_{I} \phi \ra \Int_{J}\phi$, if $I\subseteq J$
\end{itemize}

\paragraph{\cksfive} characterization of transitive closure of union in \axcdsfive,
(D\Union, T\Union, 4\Union, B\Union, 5\Union and N\Union are not
necessary in the sense that they are derivable):
\begin{itemize}
\item (K\Union)  $\Union_I(\phi\ra\psi)\ra (\Union_I\phi \ra \Union_I\psi)$
\item (D\Union) $\Union_I\phi \ra \neg {\Union_I}\neg \phi$
\item (T\Union) $\Union_I \phi \ra \phi$
\item (4\Union)  $\Union_I \phi \ra {\Union_I} {\Union_I} \phi$
\item (B\Union) $\neg\phi \ra {\Union_I} \neg {\Union_I} \phi$
\item (5\Union) $\neg{\Union_I}\phi \ra {\Union_I} \neg {\Union_I} \phi$
\item (N\Union) from $\phi$ infer $\Union_I\phi$
\item ($\Union$1) $\Union_{I}\phi \ra \Box_i (\phi \wedge {\Union_I} \phi)$, if $i\in I$
\item ($\Union$2) from $\phi\ra \bigwedge_{i\in I} \Box_i (\phi\wedge\psi)$ infer $\phi\ra \Union_I \psi$
\end{itemize}

It is known that the axiomatization $\axdsfive = \axsfive \oplus
\dksfive$ is sound and complete for the logic \ldsfive, and
$\axcdsfive = \axdsfive \oplus \cksfive$ is sound and complete for the
logic \lcdsfive (see, e.g., \cite{FHMV1995}), in the case that
$\indexes$ is finite. However, since the intersection and \unionplus
modalities are interpreted as operations over relations for standard
box operators, their properties change in accordance with those for
standard boxes. As a result, the characterization axioms and rules
vary for weaker logics. We shall look into this in the following
sections.  First we define some basic terminology that will be useful.

\begin{definition}[paths, (proper) initial segments, rest, tail]
\label{def:path}
Given a model $M=(S,R,V)$, a \emph{path} of $M$ is a finite nonempty
sequence $\ab{s_0, I_0, \ldots, I_{n-1}, s_n}$ where:
(i) $s_0, \ldots, s_n \in S$, (ii) $I_0, \ldots,
I_{n-1}$ are indices, and (iii) for all $x=0, \ldots, n-1$, $(s_x,
s_{x+1}) \in \bigcap_{i \in I}R_i$. 

Given two paths $s=\ab{s_0, {I_0}, \ldots, I_{m-1}, s_m}$ and $t=\ab{t_0, {J_0}, \ldots, J_{n-1}, t_n}$ of a model,
\begin{itemize}
\item We say $s$ is an \emph{initial segment} of $t$, denoted $s
  \preceq t$, if $m\leq n$, $s_x=t_x$ for all $x=0,\ldots,m$, and
  $I_y=J_y$ for all $y= 0,\ldots,m-1$, and we say that $t$ extends $s$
  with $\ab{J_m, t_{m+1},\ldots, J_{n-1}, t_n}$;

\item We say $s$ is a \emph{proper initial segment} of $t$, denoted $s \prec t$, if the former is an initial segment of the latter and $m<n$;

\item We write $\tail(s)$ for $s_m$, and similarly $\tail(t)$ for $t_n$;

\item When $s$ is an initial segment of $t$, we write $t\setminus s$
  to stand for the path $\ab{t_m, J_{m}, \ldots, J_{n-1}, t_n}$. Note that $\tail(s)$ is kept in $t\setminus s$, and when $s = t$, we have $t\setminus s = \ab{t_n}$. 
\end{itemize}
Given a natural number $i$, a path $s=\ab{s_0, I_0, \ldots, I_{n-1}, s_n}$ is called:
\begin{itemize}
\item An \emph{$i$-path}, if $i$ appears in all the indices of the path, i.e., $i\in \bigcap_{x=0}^{n-1} I_x$ (note that a path of length 1, such as $\ab{s_0}$, is trivially an $i$-path).
\item An \emph{$I$-path}, where $I$ is an index, if $I$ is a subset of all the indices of the path, i.e., $I\subseteq \bigcap_{x=0}^{n-1} I_x$.
\end{itemize}

\end{definition}

\section{Logics over \langd}
\label{sec:logicd}

In this section we study the logics over the language \langd, namely,
\ldk, \ldd, \ldt, \ldb, \ldsfour and \ldsfive, which means that in this section a ``formula'' stands for a formula of \langd, and a ``logic'' without further explanation refers to one of the six. We shall provide a general method for proving completeness in these logics.

The axiomatization \ax we will provide for a logic \logic is an
extension of the axiomatization for the corresponding logic without
intersection, with the characterization of intersection. The characterization of intersection is dependent on the frame conditions. For a given class of models, the characterization of intersection is listed below:
\begin{itemize}
\item $\dkk = \dkd = \{\text{K}\Int, \text{N}\Int, \Int1, \Int2\}$
\item $\dkt = \{\text{K}\Int, \text{T}\Int, \text{N}\Int, \Int1, \Int2\}$
\item $\dkb = \{\text{K}\Int, \text{T}\Int, \text{B}\Int, \text{N}\Int, \Int1, \Int2\}$
\item $\dksfour = \{\text{K}\Int, \text{T}\Int, 4\Int, \text{N}\Int, \Int1, \Int2\}$
\item $\dksfive = \{\text{K}\Int, \text{T}\Int, 5\Int, \text{N}\Int, \Int1, \Int2\}$
\end{itemize}
where \dkk is the characterization of intersection for the class of all models, \dkd for the class of all D models, \dkt for the class of all T models, and so on. We stress that D\Int is not included in \dkd: it is in fact invalid in \ldd.

By adding the characterization of intersection to the axiomatization of a logic, we get an axiomatization for the corresponding logic over \langd. To be precise, we list the axiomatizations as follows:
$$\begin{array}{rcr@{\,\oplus\,}l}
\axdk &=& \axk & \dkk\\
\axdd &=& \axd & \dkd\\
\axdt &=& \axt & \dkt\\
\axdb &=& \axb & \dkb\\
\axdsfour &=& \axsfour & \dksfour\\
\axdsfive &=& \axsfive & \dksfive\\
\end{array}$$
It is not hard to verify that all the above axiomatizations are sound in their corresponding logic, respectively.

Some of the above axiomatizations, in particular, \axdk, \axdt,
\axdsfour and \axdsfive, are given in \cite{FHMV1995}. An outline of 
proof of completeness is also found there, without details. Detailed
proofs can be found for certain special cases, such as the 
\axdsfive with only a single intersection modality for the set of all
agents (which is assumed to be finite) \cite{FHV1992}.  A general
and detailed proof based on this technique can be found in \cite{WA2013pacd}. The proof goes through an unraveling-folding
procedure, mentioned already in the introduction. Due to the subtleties
in the unraveling and folding processes, it is diffucult to apply this
technique directly to new logics, as definitions may differ already
from the beginning (for example, the
definition of a \emph{path} is already different) all the way to the very end of the
procedure. This in general becomes an obstacle for developers of new
logics with the intersection modality.

We introduce a simpler method for proving completeness, that can
easily be adapted to a range of different logics. This is a relatively
simple variant of the standard canonical model method. For each of the
logics \logic mentioned above, with corresponding axiomatization \ax, we shall show that \ax is a complete axiomatization of L, which is equivalent to finding an \logic model for every \ax-consistent set of formulas. The model we are going to build is called a \emph{standard model}.

Let $\MCS^\ax$ be the set of all maximal \ax-consistent sets of \langd-formulas.\footnote{We refer to a modal logic textbook, say \cite{BdRV2001}, for a definition of a \emph{(maximal) consistent set of formulas}.} Given \ax, given an index $I$, we shall write $\rhd_I$ to stand for a binary relation on $\MCS^\ax$, such that $\Phi \rhd_I \Psi$ iff for all $\phi$, $\Int_I\phi \in \Phi$ implies $\phi\in\Psi$. This type of relations is typically used in the definition of a canonical model (as found in modal textbooks), and are sometimes called \emph{canonical relations}. We easily get the following proposition.

\begin{proposition}[canonicity]\label{prop:cr}
For any index $I$, the canonical relation $\rhd_I$ on $\MCS^\ax$ is:
\begin{enumerate}[leftmargin=3em]
\item\label{it:d} Serial, if $I$ is singleton and \ax is \axdd;
\item\label{it:t} Reflexive, if \ax is \axdt;
\item\label{it:b} Reflexive and symmetric, if \ax is \axdb;
\item\label{it:s4} Reflexive and transitive, if \ax is \axdsfour;
\item\label{it:s5} An equivalence relation, if \ax is \axdsfive.
\end{enumerate}
We note that if \ax is \axdd, $\rhd_I$ is not necessarily serial when $I$ is not a singleton. Moreover,
\begin{enumerate}[resume,leftmargin=3em]
\item\label{it:inc} $\rhd_J \subseteq \rhd_I$, for any index $J\supseteq I$.
\end{enumerate}
\end{proposition}

\begin{definition}[canonical paths]
Given an axiomatization \ax, a \emph{canonical path} for \ax is a
sequence $\ab{\Phi_0, I_0, \ldots, I_{n-1}, \Phi_n}$ such that:

(i) $\Phi_0, \ldots, \Phi_n \in \MCS^\ax$,

(ii) $I_0, \ldots, I_{n-1}$ are indices, and

(iii) for all $x=0, \ldots, n-1$, $(s_x,s_{x+1}) \in \rhd_{I_x}$. 

\noindent Initial segments, $\tail(s)$, $i$-paths, $I$-paths, and so on, are
defined exactly like for paths in a model
(Def. \ref{def:path}).
\end{definition}

The \emph{standard models} we will define for these logics are a bit
different from the canonical model for a standard modal logic. As
mentioned existing proofs are based on transforming the canonical
model to a treelike model. We will construct a treelike model
directly: in the standard model for a logic L, a state will be a canonical path for \ax.
However, the binary relations in a standard model is dependent on the concrete logic we focus on. We now first define these binary relations and then introduce the definition of a standard model.

\begin{definition}[standard relations]
Given a logic L with its axiomatization \ax, we define
$\crelation^\logic$ as follows. For any $i \in\indexes$,
$\crelation^\logic_i$ is the binary relation on the set of canonical paths for \ax, called the \emph{standard relation for $i$}, such that:
\begin{itemize}
\item When L is \ldk or \ldd: for all canonical paths $s$ and $t$ for $\ax$, $(s,t)\in \crelation^{\logic}_i$ iff $t$ extends $s$ with $\ab{I, \Phi}$ for some $I\ni i$ and $\Phi\in\MCS^\ax$; 


\item When L is \ldt: for all canonical paths $s$ and $t$ for $\axdt$, $(s,t)\in \crelation^{\ldt}_i$ iff $t = s$ or $t$ extends $s$ with $\ab{I, \Phi}$ for some $I\ni i$ and $\Phi\in\MCS^\axdt$;

\item When L is \ldb: for all canonical paths $s$ and $t$ for $\axdb$,
  $(s,t)\in \crelation^{\ldb}_i$ iff (i) $t = s$ or (ii) $s$ extends
  $t$ with $\ab{I,\Phi}$ or (iii) $t$ extends $s$ with $\ab{I, \Phi}$ for some $i\in I$ and $\Phi\in\MCS^\axdb$; 

\item When L is \ldsfour: for all canonical paths $s$ and $t$ for
  $\axdsfour$, $(s,t)\in \crelation^{\ldsfour}_i$ iff $s$ is an
  initial segment of $t$ and $t\setminus s$ is a canonical $i$-path; 

\item When L is \ldsfive: for all canonical paths $s$ and $t$ for $\axdsfive$, $(s,t)\in \crelation^{\ldsfive}_i$ iff 
(i) $s$ and $t$ have a common initial segment $u$, and
(ii) both $s \setminus u$ and $t \setminus u$ are canonical $i$-paths. 

\end{itemize}
\end{definition}

\begin{definition}[standard models]
Given a logic L, the \emph{standard model} for L is a tuple $\cmodel^\text{L} = (\cdomain, \crelation, \cvaluation)$ such that:
\begin{itemize}
\item \cdomain is the set of all canonical paths for \ax; its elements are called \emph{states} of $\cmodel^\logic$.

\item ${\crelation} = {\crelation^\text{L}}$.

\item For any propositional variable $p$, $\cvaluation(p) = \{ s \in \cdomain \mid p \in \tail(s)\}$.
\end{itemize}
\end{definition}

\begin{lemma}[standardness]
The following hold:
\begin{enumerate}
\item $\cmodel^\ldk$ is a Kripke model;
\item $\cmodel^\ldd$ is a D model;
\item $\cmodel^\ldt$ is a T model;
\item $\cmodel^\ldb$ is a B model;
\item $\cmodel^\ldsfour$ is an S4 model;
\item $\cmodel^\ldsfive$ is an S5 model.
\end{enumerate}
\end{lemma}

\begin{lemma}[existence]
For any logic L, any state $s$ of $\cmodel^\text{L}$, and any index $I$, if $\Int_I\phi \notin \tail(s)$ then there is a state $t$ of $\cmodel^\text{L}$ such that $(s,t) \in \bigcap_{i\in I}\crelation^\text{L}_i$ and $\phi \notin \tail (t)$.
\end{lemma}
\begin{proof}
Let $s$ be a state of $\cmodel^\text{L}$ and $\Int_I\phi \notin \tail(s)$. So $\neg \Int_I\phi\in \tail(s)$. Consider the set $\Phi^- = \{\neg\phi\}\cup \{ \psi \mid \Int_I\psi \in \tail(s)\}$. We can show $\Phi^-$ is \ax consistent just as in a classical proof of the existence lemma (see, e.g., \cite{BdRV2001}). We can then extend it into a maximal consistent set $\Phi$ of formulas using the Lindenbaum construction. Since $\neg\phi\in \Phi$, $\phi\notin\Phi$. Let $t$ be $s$ extended with $\ab{I, \Phi}$. By definition it is clear that $\phi \notin \tail(t)$ and for all L, $(s,t) \in \bigcap_{i\in I}\crelation^\text{L}_i$ (since $s \crelation^\text{L}_i t$ for all $i\in I$).
\end{proof}

\begin{lemma}[truth]\label{lem:truth-d}
Given a logic \logic, a formula $\phi$, and a state $s$ of $\cmodel^\logic$,
$$\cmodel^\text{L},s\models \phi \qiff \phi \in \tail(s).$$
\end{lemma}
\begin{proof}
The proof is by induction on $\phi$. The atomic case is by definition. Boolean cases are easy to show. Interesting cases are for the modalities $\Box_i$ ($i\in\indexes$) and $\Int_I$ ($I$ is an index). We start with the case for $\Int_I\psi$.
$$\begin{array}{cll}
& \cmodel^\text{L},s \models \Int_I\psi \\
\Lra& \text{for all $t$, if $(s,t)\in\bigcap_{i\in I} \crelation^\text{L}_i$ then $\cmodel^\text{L},t\models\psi$}\\
\Lra& \text{for all $t$, if $(s,t)\in\bigcap_{i\in I} \crelation^\text{L}_i$ then $\psi\in\tail(t)$}\\
\Ra& \Int_I\psi \in \tail(s) & \text{(existence lemma)}
\end{array}$$
For the converse direction of the last step, suppose $\Int_I\psi \in
\tail(s)$ and assume towards a contradiction that there is a state $t$ such that $(s,t)\in \bigcap_{i\in I} \crelation^\text{L}_i$ and $\psi\notin\tail(t)$.
\begin{itemize}
\item If L is \ldk or \ldd, it must be that $t$ extends $s$ with $\ab{J, \Phi}$ for $J\supseteq I$ and $\Phi\in\MCS^\ax$. By definition ${\tail(s)}\rhd_J {\tail(t)}$, and by Proposition \ref{prop:cr}.\ref{it:inc}, we have ${\tail(s)}\rhd_I {\tail(t)}$.
Therefore $\psi\in\tail(t)$, which leads to a contradiction.

\item If L is \ldt, we face an extra case compared with the above, namely $s =t$. A contradiction can be reached by applying the axiom T\Int.

\item If L is \ldb, there are three cases: (i) $t = s$ or (ii) $s=\ab{t, J,\Phi}$ or (iii) $t = \ab{s, J, \Phi}$ where $J\supseteq I$ and $\Phi\in\MCS^\axb$. Case (i) can be shown similarly to the case when L is \ldt, and case (iii) to the case when L is \ldk or \ldd. For case (ii), it is important to observe that $\rhd_I$ is symmetric (Proposition \ref{prop:cr}.\ref{it:b}) and $\rhd_J \subseteq \rhd_I$ (Proposition \ref{prop:cr}.\ref{it:inc}).

\item If L is \ldsfour, $s$ must be an initial segment of $t$ and $t\setminus s$ is an $I$-path. We get ${\tail(s)}\rhd_I {\tail(t)}$ by Proposition \ref{prop:cr}.\ref{it:inc} and the reflexivity and transitivity of $\rhd_I$ (Proposition \ref{prop:cr}.\ref{it:s4}). Therefore $\psi\in\tail(t)$ which leads to a contradiction.

\item If L is \ldsfive, then $s$ and $t$ have a common initial segment $u$, and $s\setminus u$ and $t\setminus u$ are both $I$-paths. For special cases when one of $s$ and $t$ is an initial segment of the other, it can be shown like in the case when L is  \ldsfour. The interesting case is when $s$ and $t$ really fork, in this case we can show both $\tail(s) \rhd_I \tail(u)$ and $\tail(u)\rhd_I \tail(t)$ by transitivity and symmetry of $\rhd_I$ (Proposition \ref{prop:cr}.\ref{it:s5}) and Proposition \ref{prop:cr}.\ref{it:inc}, so that $\tail(s) {\rhd_I} \tail(t)$. Then $\psi\in\tail(t)$, which leads to a contradiction.
\end{itemize}
Finally, the case for $\Box_i\psi$:
\[\begin{array}{cll}
& \cmodel^\text{L},s \models \Box_i\psi \\
\Lra & \cmodel^\text{L},s \models \Int_{\{i\}}\psi& \text{(validity of $\Int1$)}\\
\Lra & \Int_{\{i\}}\psi \in \tail(s) & \text{(special case of $\Int_I\psi$)}\\
\Lra & \Box_i\psi\in \tail(s) & \text{(axiom $\Int1$)}\\
\end{array}\]
\end{proof}

\begin{theorem}[strong completeness]
Given  $\logic\in \{\ldk, \ldd, \ldt, \ldb, \ldsfour, \ldsfive\}$ and its axiomatization \ax, for any $\Phi\subseteq \langd$ and $\phi\in\langd$, if $\Phi \models \phi$, then $\Phi \vdash_{\ax}\phi$.
\end{theorem}
\begin{proof}
Suppose $\Phi \nvdash_{\ax}\phi$. It follows that $\Phi\cup\{\neg\phi\}$ is \ax consistent. Extend it to be a maximal set $\Psi$, then $\ab{\Psi}$ is a canonical path. By the truth lemma, for any formula $\psi$, we have $\cmodel,\ab{\Psi}\models \psi$ iff $\psi\in\Psi$. It follows that $\Psi$ is satisfiable, which leads to $\Phi \not\models \phi$.
\end{proof}

\section{Logics over \langcd}
\label{sec:logiccd}

In this section we study the logics with both the intersection and \unionplus modalities. The  language is set to be \langcd in this section, and by a ``logic'' without further explanation we mean one of \lcdk, \lcdd, \lcdt, \lcdb, \lcdsfour or \lcdsfive.

Compared with the characterization of intersection, that of transitive
closure of union is better behaved:
$$\ckk = \ckd = \ckt = \ckb =\cksfour = \cksfive = \{K\Union, \Union1, \Union2\}$$
These axioms are not new, as e.g., is found in \cite{FHMV1995},
although as far as we know they have not been studied as additions to
D and B in the literature. For simplicity we write \ck this set of
axioms. Additional axioms for \unionplus corresponding to specific frame conditions can be derived in specific logic systems. For instance, D\Union is a theorem of $\axd \oplus \ck$.

By adding to the axiomatization of a logic over \langd the characterization of \unionplus, we get a sound axiomatization for the corresponding logic over \langcd. To make it precise, we list the axiomatizations as follows:
$$\begin{array}{rcr@{\,\oplus\,}l}
\axcdk &=& \axdk & \ck\\
\axcdd &=& \axdd & \ck\\
\axcdt &=& \axdt & \ck\\
\axcdb &=& \axdb & \ck\\
\axcdsfour &=& \axdsfour & \ck\\
\axcdsfive &=& \axdsfive & \ck\\
\end{array}$$

In this section we will make extensive references to the names of logics and axiomatizations, and for simplicity we shall call a tuple $\sigma=\sig$ a \emph{signature}, when \logic is one of the logics \lcdk, \lcdd, \lcdt, \lcdb, \lcdsfour and \lcdsfive, \ax is the corresponding axiomatization for \logic, $\alpha$ is a formula of \langcd, and $\iota$ is an index such that every index occurring in $\alpha$ is a subset of $\iota$. Moreover, we write $\sigma(\lk)$ when the \logic in $\sigma$ is \lk, and similarly for $\sigma(\ld)$, $\sigma(\lt)$, $\sigma(\lb)$, $\sigma(\lsfour)$ and $\sigma(\lsfive)$.

\begin{definition}[closure]
Given a signature $\sigma=\sig$, the \emph{$\sigma$-closure}, denoted $cl(\sigma)$, is the minimal set of formulas satisfying the following conditions:
\begin{enumerate}
\item $\alpha\in cl(\sigma)$;
\item If $\phi\in cl(\sigma)$, then all the subformulas of $\phi$ are also in $cl(\sigma)$;
\item If $\phi$ does not start with a negation symbol and $\phi\in cl(\sigma)$, then $\neg\phi\in cl(\sigma)$;
\item For any $i\in\iota$, $\Int_{\{i\}}\phi \in cl(\sigma)$ if and only if $\Box_{i} \phi \in cl(\sigma)$;
\item For any $I$ and $J$ such that $I\subset J \subseteq \iota$, if $\Int_I\phi \in cl(\sigma)$ then $\Int_J \phi \in cl(\sigma)$;
\item\label{it:cl-union} For any $I$ and $J$ such that $I\subset J \subseteq \iota$, if $\Union_I\phi\in cl(\sigma)$ then $\{{\Int_J} {\Union_{I}} \phi \mid I\cap J \neq \emptyset\} \subseteq cl(\sigma)$.
\end{enumerate} 
\end{definition}

It is not hard to verify that $cl(\sigma)$ is finite and nonempty for any signature $\sigma$. Given $\sigma=\sig$, a set of formulas is said to be \emph{maximal \ax-consistent in $cl(\sigma)$}, if it is (i) a subset of $cl(\sigma)$, (ii) \ax-consistent and (iii) maximal in $cl(\sigma)$ (i.e., any proper superset which is a subset of $cl(\sigma)$ is inconsistent). We write $\MCS^\sigma$ for the set of all maximal \ax-consistent sets of formulas in $cl(\sigma)$.

Now we extend the canonical relations to the finitary case. Given a signature $\sigma$ and an index $I$, we may try to define a canonical relation $\rhd_I$ to be a binary relation on $\MCS^\sigma$, such that $\Phi \rhd_I \Psi$ iff for all $\phi$, $\Int_I\phi \in \Phi$ implies $\phi\in\Psi$, like we did for the logics over \langd. But there are subtleties regarding the closure. For example, transitivity may be lost in \lcdsfour, if $\Int_I\phi\in\Phi$ but ${\Int_I}{\Int_{I}}\phi\notin\Phi$ in case the latter is not included in the closure. We introduce the formal definition below.

\begin{definition}[finitary canonical relation]\label{def:fin-cr}
For a signature $\sigma=\sig$ and an index $I\subseteq \iota$, the \emph{canonical relation $\rhd_I$ for \logic} is a binary relation on $\MCS^\sigma$, such that the following hold for all $\Phi,\Psi\in\MCS^\sigma$:
\begin{itemize}
\item If \logic is \lcdk, \lcdd or \lcdt: $\Phi \rhd_I \Psi$ iff $\{\phi \mid \Int_I\phi\in\Phi\} \subseteq \Psi$;

\item If \logic is \lcdb: $\Phi \rhd_I \Psi$ iff $\{\phi \mid \Int_I\phi\in\Phi\} \subseteq \Psi$ and $\{\phi \mid \Int_I\phi\in\Psi\} \subseteq \Phi$;

\item If \logic is \lcdsfour: $\Phi \rhd_I \Psi$ iff $\{\Int_I\phi \mid \Int_I\phi\in\Phi\} \subseteq \{\Int_I\phi \mid \Int_I\phi\in\Psi\}$;%

\item If \logic is \lcdsfive: $\Phi \rhd_I \Psi$ iff $\{\Int_I\phi \mid \Int_I\phi\in\Psi\} = \{\Int_I\phi \mid \Int_I\phi\in\Psi\}$.
\end{itemize}
\end{definition}

Note that for all the logics, from $\Phi\rhd_I\Psi$ we still get that $\Int_I\phi\in\Phi$ implies $\phi\in\Psi$, as the criteria above are at least not weaker. We can get the following proposition that is similar to Proposition \ref{prop:cr}.

\begin{proposition}\label{prop:cr-2}
For any signature $\sigma=\sig$ and any index $I\subseteq \iota$, the canonical relation $\rhd_I$ on $\MCS^\sigma$ is:
\begin{enumerate}[leftmargin=3em]
\item\label{it:d-2} Serial, if $I$ is singleton and \ax is \axcdd;
\item\label{it:t-2} Reflexive, if \ax is \axcdt;
\item\label{it:b-2} Reflexive and symmetric, if \ax is \axcdb;
\item\label{it:s4-2} Reflexive and transitive, if \ax is \axcdsfour;
\item\label{it:s5-2} An equivalence relation, if \ax is \axcdsfive.
\end{enumerate}
Moreover,
\begin{enumerate}[resume,leftmargin=3em]
\item\label{it:inc-2} $\rhd_J \subseteq \rhd_I$, for any index $J$ such that $I\subseteq J\subseteq \iota$.
\end{enumerate}
\end{proposition}
\begin{proof}
For the seriality when $I=\{i\}$: given $\Phi\in\MCS^\sigma$ and a formula $\phi$ such that $\Int_{\{i\}}\phi\in \Phi$, it suffices to show the existence of a $\Psi\in \MCS^\sigma$ such that $\phi\in\Psi$. This is easy, take $\phi$ and extend it to be \ax-maximal in $cl(\sigma)$, by observing that $\phi\in cl(\sigma)$.

For reflexivity, we make use of the axiom T\Int and the fact that $cl(\sigma)$ is closed under subformulas.

For the combinations of frame conditions for the axiomatizations \axcdb, \axcdsfour and \axcdsfive, we can see that they are enforced by the definition of the canonical relation.
\end{proof}

\begin{definition}[finitary canonical paths]
Given a signature $\sigma=\sig$, a \emph{canonical path} for \ax in $cl(\sigma)$ is a sequence $\ab{\Phi_0, I_0, \ldots, I_{n-1}, \Phi_n}$ such that:

(i) $\Phi_0, \ldots, \Phi_n \in \MCS^\sigma$, 

(ii) $I_0, \ldots, I_{n-1}\subseteq \iota$, and

(iii) for all $x=0, \ldots, n-1$, $(s_x,s_{x+1}) \in \rhd_{I_x}$. 

\noindent Initial segments, $\tail(s)$, $i$-paths, $I$-paths, and so on, are
defined exactly like for paths in a model
(Def. \ref{def:path}).
\end{definition}

\begin{definition}[standard relation]\label{def:st-rel-cd}
Given a signature $\sigma=\sig$, for any $i \in \iota$, the \emph{standard relation} $\crelation^\sigma_i$ is a binary relation on the canonical paths for \ax in $cl(\sigma)$, such that:
\begin{itemize}
\item When L is \lcdk or \lcdd: for all canonical paths $s$ and $t$ for \ax in $cl(\sigma)$, $(s,t)\in \crelation^{\sigma}_i$ iff $t$ extends $s$ with $\ab{I, \Phi}$ for $\Phi\in\MCS^{\sigma}$ and some index $I$ such that $i\in I\subseteq \iota$; 

\item When L is \lcdt: for all canonical paths $s$ and $t$ for \axcdt in $cl(\sigma)$, $(s,t)\in \crelation^{\sigma}_i$ iff $t = s$ or $t$ extends $s$ with $\ab{I, \Phi}$ for $\Phi\in\MCS^{\sigma}$ and some index $I$ such that $i\in I\subseteq \iota$; 

\item When L is \lcdb: for all canonical paths $s$ and $t$ for \axcdb in $cl(\sigma)$, $(s,t)\in \crelation^{\sigma}_i$ iff (i) $t = s$ or (ii) $s=\ab{t, I,\Phi}$ or (iii) $t = \ab{s, I, \Phi}$ for $\Phi\in\MCS^{\sigma}$ and some index $I$ such that $i\in I\subseteq \iota$; 

\item When L is \lcdsfour: for all canonical paths $s$ and $t$ for
  \axcdsfour in $cl(\sigma)$, $(s,t)\in \crelation^{\sigma}_i$ iff $s$
  is an initial segment of $t$ and $t\setminus s$ is a canonical $i$-path; 

\item When L is \lcdsfive: for all canonical paths $s$ and $t$ for \axcdsfive in $cl(\sigma)$, $(s,t)\in \crelation^{\sigma}_i$ iff 
(i) $s$ and $t$ have a common initial segment $u$, and
(ii) both $s \setminus u$ and $t \setminus u$ are canonical $i$-paths. 
\end{itemize}
\end{definition}

\begin{definition}[finitary standard models]
Given a signature $\sigma=\sig$, the \emph{standard model} for $\sigma$ is a tuple $\cmodel^{\sigma} = (\cdomain, \crelation, \cvaluation)$ such that:
\begin{itemize}
\item \cdomain is the set of all canonical paths for \ax in $cl(\sigma)$. The elements of \cdomain are called \emph{states} of $\cmodel^{\sigma}$.

\item ${\crelation} = {\crelation^{\sigma}}$.

\item For any propositional variable $p$, $\cvaluation(p) = \{ s \in \cdomain \mid p \in \tail(s)\}$.
\end{itemize}
\end{definition}

\begin{lemma}[standardness]\label{lem:st-cd}
Given a signature $\sigma = \sig$, the following hold:
\begin{enumerate}
\item $\cmodel^{\sigma}$ is a Kripke model;
\item $\cmodel^{\sigma}$ is a D model when $\logic = \lcdd$ and $\ax = \axcdd$;
\item $\cmodel^{\sigma}$ is a T model when $\logic = \lcdt$ and $\ax = \axcdt$;
\item $\cmodel^{\sigma}$ is a B model when $\logic = \lcdb$ and $\ax = \axcdb$;
\item $\cmodel^{\sigma}$ is an S4 model when $\logic = \lcdsfour$ and $\ax = \axcdsfour$;
\item $\cmodel^{\sigma}$ is an S5 model when $\logic = \lcdsfive$ and $\ax = \axcdsfive$.
\end{enumerate}
\end{lemma}

\begin{lemma}[existence]
For any signature $\sigma$, any state $s$ of $\cmodel^{\sigma}$, and
any index $I \subseteq \iota$, suppose $\Int_I\phi, \Union_I\phi \in cl(\sigma)$. Then,
\begin{enumerate}
\item\label{it:ex-int} If $\Int_I\phi\notin\tail(s)$ then there is a state $t$ of $\cmodel^{\sigma}$ such that $(s,t) \in \bigcap_{i\in I}\crelation^\sigma_i$ and $\phi \notin \tail(t)$.
\item\label{it:ex-union} If $\Union_I\phi\notin\tail(s)$ then there is a state $t$ of $\cmodel^{\sigma}$ such that $(s,t) \in \biguplus_{i\in I}\crelation^\sigma_i$ and $\phi \notin\tail(t)$.
\end{enumerate}
\end{lemma}
\begin{proof}
Let $\sigma=\sig$ and $s$ be a state of $\cmodel^\sigma$. 

(\ref{it:ex-int})  Let $\Int_I\phi \notin \tail(s)$. So $\neg {\Int_I}\phi\in \tail(s)$. Consider the set $\Phi^- = \{-\phi\}\cup \{ \psi \mid \Int_I\psi \in \tail(s)\}$ (where $-\phi$ is $\psi$ if $\phi=\neg\psi$, and is $\neg\phi$ if $\phi$ is positive).  Clearly $\Phi^-\subseteq cl(\sigma)$ and it is not hard to show that it is \ax consistent. We can then extend it into a maximal consistent set $\Phi$ of formulas in $cl(\sigma)$. Since $-\phi\in \Phi$, $\phi\notin\Phi$. Let $t$ be $s$ extended with $\ab{I, \Phi}$. By definition it is clear that $\phi \notin \tail(t)$ and $(s,t) \in \bigcap_{i\in I}\crelation^\sigma_i$ (since $s \crelation^\sigma_i t$ for all $i\in I$).

(\ref{it:ex-union}) Let $\mcp$ be the property on the states of $\cmodel^\sigma$ such that for any $s$, $s\in \mcp$ iff for any $t$, if $(s,t)\in\biguplus_{i\in I}\crelation^\sigma_i$ then $\phi \in\tail(t)$. The equivalent condition is that for any state $s_0$ of $\cmodel^\sigma$, $s_0\in \mcp$ iff $\phi\in\tail(s_n)$ holds for any path $\ab{s_0,\{i_0\},\ldots,\{i_{n-1}\}, s_n}$ of $\cmodel^\sigma$ with $\{i_0,\ldots,i_{n-1}\}\subseteq I$. Let $\psi = \bigvee_{s\in\mcp} \widehat{\tail(s)}$ (where $\widehat{\tail(s)}$ is the conjunction of all formulas in $\tail(s)$). We get the following:

(a) For any $i\in I$, $\vdash_{\ax} \psi\ra \Box_i \phi$. First observe that for every $s_0\in\mcp$, any path $\ab{s_0,\{i_0\},\ldots,\{i_{n-1}\}, s_n}$ as described above is such that $\phi\in\tail(s_n)$. As a special case, for any state $s_1$, if $\ab{s_0,\{i\}, s_1}$ is a path, namely $\tail(s_0)\rhd_{\{i\}}\tail(s_1)$, then $\phi\in\tail(s_1)$. It follows that $\Box_i\phi \in \tail(s_0)$ (for otherwise it violates the first clause; just treat $\Box_i$ to be $\Int_{\{i\}}$). This means that $\Box_i\phi$ is a conjunct of every disjunct of $\psi$, and so $\vdash_{\ax} \psi\ra\Box_i\phi$.

(b) For any $i\in I$, $\vdash_{\ax} \psi\ra \Box_i \psi$. Suppose towards a contradiction that $\psi\wedge\neg\Box_i\psi$ is consistent. There must be a disjunct of $\psi$, say $\widehat{\tail(t)}$ (with $t\in\mcp$), such that $\widehat{\tail(t)}\wedge\neg\Box_i\psi$ is consistent. By properties of $\MCS^\sigma$ we have $\vdash_\ax \bigvee \{ \widehat\Phi \mid \Phi\in\MCS^\sigma\}$ (similarly $\widehat\Phi$ is the conjunction of formulas in $\Phi$).
So there must be $\Phi\in\MCS^\sigma\setminus\{\tail(s) \mid s\in\mcp\}$ such that $\widehat{\tail(t)}\wedge\neg\Box_i\neg \widehat \Phi$ is consistent. It follows that $\tail(t)\rhd_{\{i\}} \Phi$.
The path $u$ which extends $t$ with $\ab{\{i\},\Phi}$ is such that $(t,u)\in \crelation^\sigma_{\{i\}}$. Since $t\in\mcp$, we have $u\in\mcp$ as well. However, this conflicts with the fact that $\Phi\notin\{\tail(s) \mid s\in\mcp\}$.

Now we show the contraposition of the clause. Suppose $s\in\mcp$, and we must show $\Union_I\phi \in \tail(s)$. By (a) and (b) we have $\vdash_\ax \psi \ra \bigwedge_{i\in I}\Box_i(\psi\wedge\phi)$, and then by $\Union2$ we have $\vdash_\ax \psi \ra \Union_{I}\phi$. Let $\xi = \widehat{\tail (s)}$. It follows that $\vdash_{\ax} \xi\ra\psi$, as $\xi$ is one of the disjuncts of $\psi$. So we get $\vdash_\ax \xi\ra\Union_{I}\phi$, and so $\Union_I\phi \in \tail(s)$ for $\tail(s)$ is consistent.
\end{proof}

\begin{lemma}[truth]
Given a signature $\sigma$, a formula $\phi\in cl(\sigma)$, and a state $s$ of $\cmodel^{\sigma}$,
$$\cmodel^{\sigma},s\models \phi \qiff \phi \in \tail(s).$$
\end{lemma}
\begin{proof}
The proof is by induction on $\phi$. The atomic and Boolean cases are
easy to show. The cases for the modalities $\Box_i$ ($i \in \indexes$) and $\Int_I$ ($I$ is an index) are not much different from those of the proof of Lemma \ref{lem:truth-d} (we need to be careful with the closure, however; just note that all the $i$'s and $I$'s used here are bounded by an $\iota$). Here we detail the case for $\Union_I\psi$.
$$\begin{array}{cll}
& \cmodel^{\sigma},s \models \Union_I\psi \\
\Lra& \text{for all $t$, if $(s,t)\in\biguplus_{i\in I} \crelation^\sigma_i$ then $\cmodel^\sigma,t\models\psi$}\\
\Lra& \text{for all $t$, if $(s,t)\in\biguplus_{i\in I} \crelation^\sigma_i$ then $\psi\in\tail(t)$}\\
\Ra& \Union_I\psi \in \tail(s) & \text{(existence lemma)}
\end{array}$$
For the converse direction of the last step, suppose $\Union_I\psi \in \tail(s)$ and towards a contradiction that there is a state $t$ such that $(s,t)\in \biguplus_{i\in I} \crelation^\sigma_i$ and $\psi\notin\tail(t)$. So there is a path $\ab{s_0, \{i_0\}, \ldots, \{i_{n-1}\}, s_n}$ of $\cmodel^\sigma$ such that $\{i_0,\ldots,i_{n-1}\}\subseteq I$, $s=s_0$ and $t = s_{n}$.
\begin{itemize}
\item If L is \lcdk or \lcdd, it must be that $t$ extends $s$ with $\ab{J_0, \Phi_1, \ldots, J_{n-1}, \Phi_n}$ where $\psi\notin\Phi_n$ and for each $x$, $i_x \in J_x$ and $\Phi_x\in\MCS^{\sigma}$. By definition ${\tail(s_0)\rhd_{J_0} \Phi_1 \rhd_{J_1} \cdots \rhd_{J_{n-1}} \Phi_n}$. By the axioms $\Union1$, $\Int1$ and $\Int2$ we can get $\vdash_\ax \Union_{I}\psi \ra {\Int_{J_0}}{\Union_{I}}\psi$, and since ${\Int_{J_0}}{\Union_{I}}\psi\in cl(\sigma)$ we have ${\Union_{I}}\psi \in\Phi_1$. Carrying this out recursively we get ${\Union_{I}}\psi \in\Phi_n$ and so $\psi\in\Phi_n$ by T\Union, which contradicts $\psi\notin\tail(t)$.

\item If L is \lcdt, we face an extra case compared with the above, namely $s =t$. A contradiction can be achieved by applying the axiom T\Union.

\item If L is \lcdb, there are three cases: (i) $s_{x+1} = s_x$ or (ii) $s_{x}=\ab{s_{x+1}, J,\Phi}$ or (iii) $s_{x+1} = \ab{s_x, J, \Phi}$ where $J\supseteq I$ and $\Phi\in\MCS^{\sigma}$. In all cases, by similar reasoning to the above (for case (ii) we use the symmetric condition for $\rhd_I$), we can show that $\psi\in \tail(s_{x+1})$ given $\Union_I\psi\in\tail(s_x)$, and then reach a contradiction similarly.

\item If L is \lcdsfour, $s_x$ ($0\leq x<n$) must be an initial
  segment of $s_{x+1}$ and $s_{x+1}\setminus s_x$ is a finitary canonical $i_x$-path ($i_x\in I$). By the axioms $\Union1$ and $\Int1$, $\vdash_{\axcdsfour} \Union_I\psi \ra \Int_{\{i_x\}} {\Union_I} \psi$.  So we get ${\Union_I}\psi\in\tail(s_{x+1})$ (we use T\Union in the case when $s=t$). Recursively carrying this out, we get $\Union_I\psi\in\tail(t)$, and so $\psi\in\tail(t)$ which leads to a contradiction.

\item If L is \lcdsfive, then $s_x$ and $s_{x+1}$ have a common
  initial segment $u$, and $s_x\setminus u$ and $s_{x+1}\setminus u$
  are both finitary canonical $i_x$-paths. Since $\vdash_{\axcdsfive} \Union_I\psi \ra \Int_{\{i_x\}}\psi$ (for all $i_x$), $\Int_I\psi\in s_0$, and by the definition of $\rhd_I$, $\Int_I \psi\in s_x$ for all $x$, so  $\psi\in\tail(t)$ which leads to a contradiction as well.
\end{itemize}
\end{proof}

\begin{theorem}[weak completeness]
Let \ax be the corresponding axiomatization introduced for a logic $\logic\in\{\lcdk, \lcdd, \lcdt, \lcdb, \lcdsfour, \lcdsfive\}$. For any $\phi\in\langcd$, if $\models \phi$, then $\vdash_{\ax}\phi$.
\end{theorem}
\begin{proof}
Suppose $\nvdash_{\ax}\phi$. It follows that $\{\neg\phi\}$ is \ax consistent. Extend it to be a maximal set $\Phi$ in $cl((\logic,\ax,\neg\phi,\iota))$ with $\iota$ the union of all the indices occurring in $\phi$, then $\ab{\Phi}$ is a canonical path for \ax in $cl((\logic,\ax,\neg\phi,\iota))$. By the truth lemma, for any formula $\psi$, we have $\cmodel^{(\logic,\ax,\neg\phi,\iota)},\ab{\Phi}\models \phi$ iff $\phi\in\Phi$. It follows that $\Phi$ is satisfiable, which leads to $\not\models \phi$.
\end{proof}

\section{Discussion}
\label{sec:con}

We focused mainly on the completeness proof for the modal logics, K,
D, T, B, S4 and S5, extended with intersection and with or without the
transitive closure of union, but the method applies to many canonical
multi-modal logics (including many of those normal modal logics
between K and S5) with the intersection modality. By avoiding the
model translation processes used in the unraveling-folding method and
building a standard model directly, the proofs we present are
dramatically simpler than those found in the literature. We believe
that the readers who are familiar with the canonical model method for
completeness proofs of modal logics will find the proofs very
familiar and straightforward.

While our approach is inspired by simplifying the existing proof
technique, the standard model we build is not identical to the model
produced by the unraveling-folding processes: it is simpler because we
do not have to use so-called reductions of paths. We emphasize,
however, that the unraveling-folding method was still important for us
to arrive at this proof technique:
it explains why we take (finitary) canonical paths to be the states of the
standard model. Further work that could be interesting is to show the bisimilarity of the model we build to that by the unraveling-folding processes.

We omit the details here, but our method can be applied directly to many other
logics extended with intersection modalities, including popular
systems of epistemic and doxastic logics such as S4.2, S4.3, S4.4 and
KD45.

Finally, it is worth mentioning that our proof technique is slightly
more general than existing proofs in that it allows a (countably)
infinite set of boxes. This slightly complicates the proofs in the
cases with transitive closure of the union, requiring the use of the
$\sigma$ signatures.

\bibliographystyle{splncs04}
\bibliography{gbel}

\begin{thebibliography}{10}
\providecommand{\url}[1]{\texttt{#1}}
\providecommand{\urlprefix}{URL }
\providecommand{\doi}[1]{https://doi.org/#1}

\bibitem{asia2015}
{\AA}gotnes, T., Alechina, N.: Embedding coalition logic in the minimal normal
  multi-modal logic with intersection. In: Ono, H., Ju, S. (eds.) Proceedings
  of the Second Asian Workshop on Philosophical Logic, Logic in Asia: Studia
  Logica Library, vol.~1. Springer International Publishing (2015)

\bibitem{BCMNP-S2017}
Baader, F., Calvanese, D., McGuinness, D.L., Nardi, D., Patel-Schneider, P.F.:
  The Description Logic Handbook: Theory, Implementation, and Applications.
  Cambridge University Press, 2 edn. (2017)

\bibitem{BHLS2017}
Baader, F., Horrocks, I., Lutz, C., Sattler, U.: An Introduction to Description
  Logic. Cambridge University Press (2017)

\bibitem{BdRV2001}
Blackburn, P., de~Rijke, M., Venema, Y.: Modal Logic, Cambridge Tracts in
  Theoretical Computer Science, vol.~53. Cambridge University Press, Cambridge
  University Press (2001)

\bibitem{Chellas1980}
Chellas, B.F.: Modal Logic: An Introduction. Cambridge University Press (1980)

\bibitem{FHV1992}
Fagin, R., Halpern, J.Y., Vardi, M.Y.: What can machines know? on the
  properties of knowledge in distributed systems. J.~ACM  \textbf{39}(2),
  328--376 (1992)

\bibitem{FHMV1995}
Fagin, R., Halpern, J.Y., Moses, Y., Vardi, M.Y.: Reasoning about Knowledge.
  {Cambridge, MA: The MIT Press} (1995), hardcover edition

\bibitem{HM1992}
Halpern, J.Y., Moses, Y.: A guide to completeness and complexity for modal
  logics of knowledge and belief. Artificial Intelligence  \textbf{54},
  319--379 (1992)

\bibitem{HKT2000}
Harel, D., Kozen, D., Tiuryn, J.: Dynamic Logic. The MIT Press (2000)

\bibitem{hoek92making}
van~der Hoek, W., Meyer, J.J.C.: Making some issues of implicit knowledge
  explicit. International Journal of Foundations of Computer Science
  \textbf{3}(2),  193--224 (1992)

\bibitem{MvdH1995}
Meyer, J.J.C., van~der Hoek, W.: Epistemic Logic for {AI} and Computer Science.
  Cambridge University Press, Cambridge, England (1995)

\bibitem{Wang2013thesis}
W{\'{a}}ng, Y.N.: Logical Dynamics of Group Knowledge and Subset Spaces. Ph.D.
  thesis, University of Bergen (2013)

\bibitem{WA2013pacd}
W\'ang, Y.N., {\AA}gotnes, T.: Public announcement logic with distributed
  knowledge: Expressivity, completeness and complexity. Synthese
  \textbf{190}(1 supplement),  135--162 (2013)

\bibitem{WA2015rck-del}
W\'{a}ng, Y.N., {\AA}gotnes, T.: Relativized common knowledge for dynamic
  epistemic logic. Journal of Applied Logic  \textbf{13}(3),  370--393 (2015).
  \doi{10.1016/j.jal.2015.06.004}

\end{thebibliography}

\end{document}